\documentclass[11pt]{article}

\usepackage[english]{babel}

\usepackage[dvips]{color}
\usepackage{xcolor}

\usepackage{graphicx}
\usepackage{stmaryrd}
\usepackage{latexsym}
\usepackage{fancyhdr}
\usepackage{comment}
\usepackage{indentfirst}
\usepackage{supertabular}
\usepackage{tabularx}
\usepackage{rotating}
\usepackage{lscape} 

\usepackage{amsmath,amsfonts,amssymb,amsthm,mathtools}
\usepackage{latexsym}

\usepackage{algorithmicx,algorithm}
\usepackage[noend]{algpseudocode}
\usepackage{breqn}
\floatname{algorithm}{Algorithm}
\algblock{Input}{EndInput}
\algnotext{EndInput}
\algblock{Output}{EndOutput}
\algnotext{EndOutput}
\algblock{Data}{EndData}
\algnotext{EndData}
\algnewcommand{\algorithmicgoto}{\textbf{go to}}
\algnewcommand{\GoTo}[1]{\algorithmicgoto~\ref{#1}}
\newcommand{\Desc}[2]{\State \makebox[2em][l]{#1}#2}

\usepackage {tikz}
\usetikzlibrary {positioning}

\usepackage{ulem}
\usepackage{authblk}

\usepackage{a4wide}

\usepackage[utf8]{inputenc}

\title{Decoding Reed-Solomon codes by solving a bilinear system with a Gr\"obner basis approach}

\newcommand{\F}{\mathbb{F}} \newcommand{\Fq}{\F_q}

\newcommand{\ba}{\mathbf{a}}
\newcommand{\bb}{\mathbf{b}}

\newcommand{\cE}{\mathcal{E}}

\newcommand{\cI}{\mathcal{I}}

\newcommand{\cM}{\mathcal{M}}

\newcommand{\cS}{\mathcal{S}}

\DeclareMathOperator{\coeffv}{\operatorname{coeff}}
\newcommand{\RS}{\mathbf{RS}}

\newcommand{\coeff}[2]{\coeffv\left(#1,#2\right)}

\newcommand{\LM}{LM}

\renewcommand{\deg}[1]{\operatorname{deg}\left(#1\right)}

\newcommand{\eqdef}{\stackrel{\text{\mbox{\tiny def}}}{=}}
\newcommand{\vect}[2]{\langle#1\rangle_{#2}} 
\newcommand{\power}{\pow{1}}
\newcommand{\pow}[1]{q_{#1}}
\newcommand{\supp}{a}
\newcommand{\suppv}{\ba}
\newcommand{\rec}{b}
\newcommand{\recv}{\bb}
\newcommand{\sol}{r}
\newcommand{\Mac}{\cM^{\text{\mbox{\tiny acaulay}}}}
\newcommand{\IInt}[2]{\llbracket #1, #2 \rrbracket}
\newcommand{\modulo}[2]{\left[#1\right]_{#2}}
\newcommand{\high}[1]{#1_{H}}
\newcommand{\inc}{\in_{\textrm{{\scriptsize{coef}}}}}

\theoremstyle{plain}
\newtheorem{Th}{Theorem}
\newtheorem{theorem}[Th]{Theorem}
\newtheorem{lemma}[Th]{Lemma}
\newtheorem{proposition}[Th]{Proposition}
\newtheorem{fact}[Th]{Fact}

\newtheorem{Remark}[Th]{Remark}

\begin{document}

\author[1,3]{Magali Bardet \thanks{magali.bardet@univ-rouen.fr}}
\author[2,3]{Rocco Mora \thanks{rocco.mora@inria.fr}}
\author[3]{Jean-Pierre Tillich \thanks{jean-pierre.tillich@inria.fr}}
\affil[1]{LITIS, University of Rouen Normandie}
\affil[2]{Sorbonne Universit\'es, UPMC Univ Paris 06}
\affil[3]{Inria, Team COSMIQ,
    2 rue Simone Iff, CS 42112,
    75589 Paris Cedex 12, France}
    \date{}
\maketitle

\begin{abstract} Decoding a Reed-Solomon code can be modeled by a bilinear system which can be solved by Gr\"obner basis techniques. We will show that in this particular case, these techniques are much more efficient than for generic bilinear systems with the same number of unknowns and equations (where these techniques have exponential complexity). Here we show that they are able to solve the problem in polynomial time up to the Sudan radius. Moreover, beyond this radius these techniques recover automatically polynomial identities that are at the heart of improvements of the power decoding approach for reaching the Johnson decoding radius. They also allow to derive new polynomial identities that can be used to derive new algebraic decoding algorithms for Reed-Solomon codes. We provide numerical evidence that this sometimes allows to correct efficiently slightly more errors than the Johnson radius.  
\end{abstract}

\section{Introduction}

\par{\bf Decoding a large number of errors in Reed-Solomon codes.}
A long-standing open problem in algebraic coding theory
 was that of decoding Reed-Solomon codes beyond the
error-correction radius, $
\frac{1-R}{2}$ (where $R$ stands for the code rate). This problem was solved in a breakthrough paper by Sudan in \cite{S97b} where it was shown that there exists an algebraic decoder that works up to a fraction of errors $1-\sqrt{2R}$ (the so called Sudan radius here). This was even improved later on by Guruswami and Sudan in \cite{GS98} with a decoder that works up to the Johnson radius $1-\sqrt{R}$. This represents in a sense the limit for such decoders since these decoders are list decoders that output all codewords up to this radius and beyond this radius the list size is not guaranteed to be polynomial anymore. However, if we do not insist on having a decoder that outputs all codewords within a certain radius, or if we just want  a decoder that is successful most of the time on the $q$-ary symmetric channel of crossover probability $p$, then we can still hope to have an efficient decoder beyond this bound. Moreover, it is even interesting to investigate if there are decoding algorithms of say subexponential complexity above the radius $1-\sqrt{R}$. 

\medskip
\par{\bf A Gr\"obner basis approach.} Our approach for decoding is to model  decoding by an algebraic system and then to solve it with Gröbner bases techniques. At first sight, it might seem that  this approach is not new in this setting: such techniques  have already been used here, mainly to solve algebraic systems involved in the Guruswami-Sudan approach~\cite{LO06b,LO08,MM11,T10,ZS10,BW18}. They were used up to now on systems where such techniques are expected to run efficiently just because the number of variables was very small for instance: for instance~\cite{LO06b,LO08,MM11,T10} consider only two variables $X$ and $Y$ corresponding to the variables of the interpolation polynomial which is sought. 

Our approach in this paper is different. We consider the classical
bilinar system~\eqref{eq:bilinear} modeling the decoding problem. This is at first sight a no go, because solving generic bilinear systems with the same number of
variables and equations  is here of
exponential complexity. However it will turn that for the system at hand this approach is astonishingly efficient: we will for instance show that it runs in polynomial time when the fraction of errors is below the
Sudan radius.
Indeed, consider a $k$-dimensional Reed-Solomon code of length $n$ over $\F_q$ with support $\suppv = (\supp_i)_{1 \leq i \leq n} \in \F_q^n$:
$$\RS_k(\suppv)
=\{(P(\supp_i))_{1 \leq i \leq n}:\;P \in \F_q[X],\;\deg P < k\}.$$ Let  $\recv =(\rec_i)_{1 \leq i \leq n}$ be the received word,
$\cE$ be the set of positions  in error and define the error locator as usual
\begin{equation}
\label{eq:error_locator}
\Lambda(X) \eqdef \Pi_{i \in \cE} (X-a_i).
\end{equation}
From this, we can write the bilinear system with unknowns the coefficients $p_i$ of the polynomial 
$P(X) = \sum_{i=0}^{k-1} p_i X^i$ corresponding to the codeword that was sent and the coefficients $\lambda_j$ of the error locator polynomial $\Lambda(X) = X^t+ \sum_{j=0}^{t-1} \lambda_j X^j$ if we assume that there were $t$ errors. We have $n$ bilinear equations in the $k+t$ variables $p_i$'s and $\lambda_j$'s coming from the $n$ relations
$
P(\supp_\ell)\Lambda(\supp_\ell) = \rec_\ell \Lambda(\supp_\ell), \;\;\ell \in \IInt{1}{n}, 
$
namely 
\begin{equation}
\label{eq:bilinear}
\sum_{i=0}^{k-1}\sum_{j=0}^{t} \supp_\ell^{i+j} p_i\lambda_j = \sum_{j=0}^{t} \rec_\ell \supp_\ell^j \lambda_j,\;\;\ell \in \IInt{1}{n} \;\;\text{and $\lambda_t=1$.}
\end{equation}

\par{\bf Gr\"obner basis techniques: a simple and automatic way for obtaining a polynomial time  algorithm in our case.} Standard Gr\"obner bases techniques can be used to solve this system, however solving \eqref{eq:bilinear} is much easier than solving a generic bilinear system with the same parameters. In particular these techniques solve typically in polynomial time the decoding problem when the fraction of errors is below the Sudan radius. This is explained  in Section \ref{ss:Sudan}. The reason why the Gr\"obner basis approach works in polynomial time is related to power-decoding \cite{SSB10,N14} and can be explained by similar arguments. However, the nice thing about this Gr\"obner basis approach is that the algorithm itself is very simple and can be given without any reference to power decoding (or the Sudan algorithm). 
The computation of the Gr\"obner basis reveals degree falls which are instrumental for its very low complexity. Understanding these degree falls can be explained by the polynomial equations used by power decoding.  However, this simple algorithm also appears to be very powerful beyond the Sudan bound: experimentally it seems that it is efficient up to the Johnson radius and that it is even able to correct more errors in some cases than the refinement of the original power decoding algorithm \cite{N18} (which reaches asymptotically the Johnson radius).
This is demonstrated in Section \ref{sec:experiments}. 

\par{\bf Understanding the nice behavior of the Gr\"obner basis approach.}
Moreover, 
trying to understand theoretically why this algorithm behaves so well, is not only explained by the polynomial relations which are at the heart of the power decoding approach, it also reveals new polynomial relations that are not exploited by the power decoding approach as shown in Section \ref{sec:explanation}. In other words, this approach not only gives an efficient algorithm, it also exploits other polynomial relations.
It seems fruitful to understand and describe them, this namely paves the road towards new algebraic decoders of Reed-Solomon codes.

\par{\bf Notation.} Throughout the paper we will use the following notation.
The integer interval $\{a,a+1,\cdots,b\}$ will be denoted by $\IInt{a}{b}$. For a polynomial $Q(X)=\sum_{i=0}^m q_i X^i$, $\coeff{Q(X)}{X^s}$ stands for the coefficient $q_s$ of $X^s$ in $Q(X)$. For two polynomials $Q(X)$ and $G(X)$, $\modulo{Q(X)}{G(X)}$ stands for the remainder of $Q(X)$ divided by $G(X)$.

\section{The Algorithm}\label{sec:algorithm}

Consider an algebraic system of equations 
\begin{equation}
\label{eq:system}
\left\{ \begin{array}{lcr} f_1(x_1,\cdots,x_\ell)&=&0 \\
&\vdots  & \\
f_m(x_1,\cdots,x_\ell)&=&0
\end{array}
\right.
\end{equation}
where the $f_i$'s are polynomials in $x_1,\cdots,x_\ell$. Such systems can be solved by Gr\"obner basis techniques (see \cite{CLO15} for instance). To simplify the discussion assume that we have a unique solution to the algebraic system \eqref{eq:system} and that the polynomial ideal $\cI$ generated by the $f_i$'s is radical, meaning that whenever there is a polynomial $f$ and a positive integer $s$ such that $f^s$ is in $\cI$, then  $f$ is in $\cI$. This
seems to be the typical case for~\eqref{eq:bilinear} when the number of errors  is below the Gilbert-Varshamov bound. In such a case, the reduced Gr\"obner basis of the ideal $\cI$ is
given by the set $\{x_1-\sol_1,\cdots,x_n-\sol_\ell\}$ for any admissible monomial ordering, where $(\sol_1,\cdots,\sol_\ell)$ stands for the unique solution of \eqref{eq:system} (this is standard, see for instance \cite[Lemma 2.4.3, p.40]{B04}). Recall here that a Gr\"obner basis of a polynomial ideal is defined for a given admissible monomial ordering\footnote{This is a total ordering $<$ of the monomials such that (i) $m < m' \implies mt < m't$  for any monomial $t$ (ii) every subset of monomials has a smallest element.} as a generating set $\{g_1,\cdots,g_s\}$ of the ideal such that the 
ideal generated by the leading monomials  $\LM(g_i)$ (where  $\LM(g)$ is the largest monomial in  $g$) of the $g_i$'s coincides with the ideal generated by all the leading monomials of the elements of $\cI$:
$$
\langle \LM(g_1),\cdots,\LM(g_s) \rangle = \langle \LM(f): f \in \cI\rangle.
$$

 We will adopt Lazard's point of view \cite{L83} to compute a Gr\"obner basis and use Gaussian elimination on the Macaulay matrices associated to the system.
 The main known and efficient algorithm for this is Faug\`ere's F4 algorithm~\cite{F99}, see ~\cite{CLO15} for background on the subject.

We
recall that the Macaulay matrix $\Mac_D(\cS)$ in degree $D$ of a set $\cS=\{f_1,\cdots,f_m\}$ of polynomials  is the matrix whose columns correspond to the
monomials of degree $\leq D$ sorted in descending order w.r.t.~a chosen
monomial ordering, whose rows correspond to the polynomials $t f_i$ for
all $i$ where $t$ is a monomial of degree $\leq D-\deg{f_i}$, and whose
entry in row $tf_i$ and column $u$ is the coefficient of the monomial
$u$ in the polynomial $tf_i$. A Gr\"obner basis for the system can be computed by computing a row echelon form
of $\Mac_D$ for large enough $D$ \cite{L83} and~\cite[chap. 10]{CLO15}. However, this way of solving \eqref{eq:bilinear} is very inefficient 
(unless $t \leq \frac{n-k}{2}$ where direct row echelonizing \eqref{eq:bilinear} is enough) because during the Gaussian elimination process we have a sequence of degree falls which are instrumental for computing a Gr\"obner basis by staying at a very small degree (this appears clearly if we use for instance  Faug\`ere's F4 algorithm~\cite{F99} on 
\eqref{eq:bilinear}).

A {\em degree fall} is a polynomial combination $\sum_{i=1}^m g_i f_i$ of the $f_i$'s which satisfies
$$
0< s \eqdef \deg{\sum_{i=1}^m g_i f_i} < \max_{i=1}^m
\deg{g_i f_i}.
$$
We say that $\deg{\sum_{i=1}^m g_i f_i} $ is a degree fall of {\em degree} $s$. 

The simplest example of such a degree fall occurs in \eqref{eq:bilinear} when $t < n-k$. Here there are linear combinations of the
bilinear equations of \eqref{eq:bilinear} giving linear equations. This can be verified by  performing the change of variables $z_s \eqdef \sum_{i,j:i+j=s} p_i \lambda_j$ in \eqref{eq:bilinear} and get
the system
\begin{equation}
\label{eq:bilinearprime}
\sum_{s=0}^{t+k-1} \supp_\ell^{s} z_s = \sum_{j=0}^{t} \rec_\ell \supp_\ell^j \lambda_j,\;\;\ell \in \IInt{1}{n}.
\end{equation}
In other words, by eliminating the $z_s$'s in these equations we obtain linear equations involving only the $\lambda_i$'s.
When $t \leq \frac{n-k}{2}$ there are enough 
such equations to recover from them the $\lambda_i$'s and by substituting for them in \eqref{eq:bilinear} the $p_i$'s by solving again a linear system. Of course, this is well known, and there are much more efficient algorithms for solving this system but still it is interesting to notice that the Gr\"obner basis approach already yields a polynomial time algorithm for the particular bilinear system \eqref{eq:bilinear} despite being exponential (for a large range of parameters) for generic bilinear systems with the same number of unknowns and equations as \eqref{eq:bilinear} \cite{FSS11,S12}.

A slightly less trivial degree fall behavior is obtained in the case the fraction of errors is Sudan's radius.  Here, after substituting for the $\lambda_i$'s which can be expressed as linear functions of the other $\lambda_i$'s by using the aforementioned linear equations involving the $\lambda_i$'s we obtain  new bilinear equations $f'_1,\cdots,f'_m$. It turns out that we can perform linear combinations on these $f'_i$'s to eliminate the monomials of degree $2$ in them and  derive new linear equations involving only the $\lambda_i$'s. This is proved in Subsection \ref{ss:Sudan}. This process can be iterated and there are typically enough such linear equations to recover the $\lambda_i$'s in this way as long as $t$ is below or equal to the Sudan decoding radius. As explained above, this allows to recover the right codeword by plugging the values for $\lambda_i$ in \eqref{eq:bilinear} and solving the corresponding linear system in the $p_i$'s. 

Note that here, and in all the paper, we are considering \emph{graded}
monomial orderings (a monomial of degree $d$ is always smaller than a
monomial of degree $d' >d$).  Through this paper, we use the notion of
affine $D$-Gr\"obner basis, which is the truncated Gr\"obner basis
obtained by ignoring computations in degree greater than $D$. It is
well known that there exists a $D$ such that a $D$-Gr\"obner basis is
indeed a Gr\"obner basis.  We describe here 
Algorithm~\ref{algo:simpleGrobner} which computes a $D$-Gr\"oner basis
of a given system through linear algebra. It is less efficient than standard algorithms but has the merit of being simple and showing what is computed during such algorithms. It is also of polynomial time complexity when $D$ is fixed. It uses the function
 $\texttt{Pol}(M)$ that returns the
polynomials represented by the rows of a Macaulay matrix  $M$. 

\begin{algorithm}
\caption{$D$-Gröbner Basis} \label{algo:simpleGrobner} 
\begin{algorithmic}
\Input
\Desc{$D$}{Maximal degree,}
\Desc{$\mathcal S$ }{$= \{f_1,\cdots,f_m\}$ set of polynomials.}
\EndInput
\Repeat 
\State{$\cS \gets \texttt{Pol}(\texttt{EchelonForm}(\Mac_D(\cS)))$}
\Until{$\dim_{\mathbb F_q}\cS$ has not increased.}
\State{Output $\cS$.}
\end{algorithmic}
\end{algorithm}

It is clear that Algorithm~\ref{algo:simpleGrobner} terminates and has
a polynomial complexity if $D$ is fixed. The previous remarks show
that we can decode up to the Sudan decoding radius with
$D=2$.  However, when the number of errors becomes bigger, $D=2$ is not
enough to exhibit more degree falls. We have to go  a higher degree.
However, already  taking $D=3$ yields  interesting degree falls that are instrumental to the generalization of the power decoding approach of \cite{N18} decoding up to the Johnson radius.

\section{A partial explanation of the algebraic behavior}\label{sec:explanation}
\subsection{Correcting up to the Sudan bound in polynomial time}\label{ss:Sudan}
The efficiency of Algorithm \ref{algo:simpleGrobner} is already demonstrated by the fact that choosing 
$D=2$ in it corrects in polynomial time as many errors as Sudan's algorithm. Choosing $D=2$ in Algorithm \ref{algo:simpleGrobner} means that we just keep the equations of degree $2$ and try to produce new linear equations by  linear combinations of the equations of degree $2$ aiming at eliminating the degree $2$ monomials. 
 The efficiency of this 
algorithm is related to power decoding \cite{SSB10}: the algorithm finds automatically the linear equations exploited by the power decoding approach. It is here convenient in order to explain the effectiveness of  the Gr\"obner basis approach to bring in an equivalent algebraic system which is basically the key equation implicit in Gao's decoder \cite{G03a} (and the one used in the power decoding approach) which is the following polynomial
equation:
\begin{equation}
\label{eq:bilinear_better}
 P(X)\Lambda(X) \equiv R(X)\Lambda(X) \mod G(X)
\end{equation}
where $R(X)$ is the polynomial of degree $\leq n-1$ interpolating the received values, i.e
$$
R(a_\ell)=\rec_\ell,  \;\;\ell \in \IInt{1}{n}\;\;\text{and } G(X) \eqdef \Pi_{\ell=1}^n (X-\supp_\ell).
$$
Note that these two polynomials are immediately computable by the receiver (and $G$ can be precomputed).
By using the same unknowns as in \eqref{eq:bilinear}, namely the coefficients of $P(X)$ and $\Lambda(X)$ we obtain a bilinear system with $n$ equations. It is readily seen that 
\begin{proposition}
The bilinear systems \eqref{eq:bilinear} and \eqref{eq:bilinear_better} are equivalent: \eqref{eq:bilinear_better} can be obtained from linear combinations of \eqref{eq:bilinear} and vice versa.
\end{proposition}
This follows on the spot from 
\begin{fact}
For any polynomial $Q(X)  \in \Fq[X]$ of degree $<n$, the coefficients of $Q$ can be expressed as linear combinations of $Q(\supp_1),\cdots,Q(\supp_n)$.
\end{fact}
This  fact is just a consequence that $Q$ coincides with its  interpolation polynomial on the points $(a_\ell,Q(a_\ell))$ and that this interpolation polynomial is given by
$$
Q(X) = \sum_{\ell=1}^n Q(\supp_\ell) \frac{\Pi_{j \neq \ell} (X-a_j)}{\Pi_{j \neq \ell} (\supp_\ell - \supp_j)}.
$$
To understand now why \eqref{eq:bilinear_better} can be derived from \eqref{eq:bilinear}, we just notice that if we bring in
\begin{eqnarray*}
Q(X) & \eqdef & P(X)\Lambda(X)-R(X)\Lambda(X)\\
S(X) & \eqdef & Q(X) \mod G(X),
\end{eqnarray*}
then 
\begin{itemize}
\item \eqref{eq:bilinear} amounts to write $Q(a_\ell)=0$ for $\ell$ in $\IInt{1}{n}$  and to express the $Q(\supp_\ell)$'s as quadratic forms in the $\lambda_i$'s and the $p_j$'s.
\item Since $Q(\supp_\ell)=S(\supp_\ell)$ for all $\ell$ in $\IInt{1}{n}$ and since $S$ is of degree $< n$ we can use the previous fact and express its coefficients linearly in terms of the $S(\supp_\ell)=Q(\supp_\ell)$'s.
\item Since \eqref{eq:bilinear_better} is nothing but  expressing that the coefficients of $S(X)$ are all equal to $0$, we obtain that the equations of \eqref{eq:bilinear_better} can be obtained from linear combinations of the equations of 
\eqref{eq:bilinear}.
\end{itemize}

Conversely since $S(\supp_\ell)$ can be written as a linear combination of the coefficients of $S(X)$, the quadratic equations in the $\lambda_i$'s and the $p_i$'s obtained by writing $S(\supp_\ell)=0$ are linear combinations of the quadratic equations given by \eqref{eq:bilinear_better}. These equations $S(\supp_\ell)=0$ coincide with the 
equations in \eqref{eq:bilinear}, since $Q(\supp_\ell)=S(\supp_\ell)$ for all $\ell$ in $\IInt{1}{n}$.

The point of \eqref{eq:bilinear_better} is that 
\begin{itemize}
\item These equations are more convenient to work with to understand what is going on algebraically during the Gr\"obner basis calculations of Algorithm \ref{algo:simpleGrobner}.
\item They give directly $n-k-t+1$ linear  equations, since (i) the coefficient of $S(X)$ of degree $d \in \IInt{t+k}{n-1}$ coincides with the coefficient of the same degree in $-R(X)\Lambda(X) \mod G(X)$ since 
$\Lambda(X)P(X)$ is of degree $\leq t+k-1$; (ii) the coefficient of $S(X)$ of degree $t+k-1$ is equal to 
$p_{k-1}-\coeff{\modulo{\Lambda(X)R(X)}{G(X)}}{X^{t+k-1}}$ because $\Lambda(X)$ is monic and of degree $t$.
\end{itemize}
With this at hand we can now prove that 
\begin{proposition} \label{prop:S_1_elim}
Let $\power\eqdef \max\{u:t+(k-1)u \leq n-1\}=\left\lfloor\frac{n-t-1}{k-1}\right\rfloor$.
All affine functions in the $\lambda_i$'s  
of the form $\coeff{\modulo{\Lambda(X)R^j(X)}{G(X)}}{X^u}$ for $j \in \IInt{1}{\power}$  and $u \in \IInt{t+(k-1)j+1}{n-1}$ are in the linear span of the set $\cS$ output by Algorithm \ref{algo:simpleGrobner} when
$D=2$.
\end{proposition}

\begin{Remark}\label{rem:easy} The fact that these are indeed affine functions follows on the spot from generalizing the degree considerations above: $\Lambda(X)P(X)^j$ is of degree $\leq t+(k-1)j$. 
\end{Remark}

\begin{proof}
Notice that from the equivalence we have just proved,  the space generated by $\cS$ contains initially (and therefore all the time) the space of affine functions in the $\lambda_i$'s   generated by 
$$
\coeff{\modulo{-\Lambda(X)R(X)}{G(X)}}{X^u}=\coeff{\modulo{\Lambda(X)P(X)-\Lambda(X)R(X)}{G(X)}}{X^u}, \;\text{ for all $u \in \IInt{t+k}{n-1}$.}
$$

Now proceed by induction on $j$, and assume that at some point the space generated by $\cS$ contains the linear span of  the  affine functions 
$$\coeff{\modulo{-\Lambda(X)R^j(X)}{G(X)}}{X^u}=\coeff{\modulo{\Lambda(X) P(X)^j-\Lambda(X) R(X)^j}{G(X)}}{X^u},$$
for all  $u \in \IInt{t+(k-1)j+1}{n-1}$ where $j$ is some integer in the interval
$\IInt{1}{\power-1}$.
Note that
\begin{align}
\label{term:start}   &\left( \Lambda P^{j+1}-\Lambda R^{j+1}\right) \mod G \\
 \nonumber  =& \left( P(\Lambda P^j-\Lambda R^j) + R^j(\Lambda P-\Lambda R)\right) \mod G \\
 \label{term:end}  =& \left(P(\Lambda P^j-\Lambda R^j \mod G) + R^j(\Lambda P-\Lambda R \mod G)\right) \mod G.
\end{align}
We use the equality between the polynomials \eqref{term:start} and \eqref{term:end} to claim that their coefficients 
should coincide for all the degrees $\IInt{t+(j-1)(k-1)}{n-1}$. 
Note now that after the elimination of variables performed so far,  this makes that all coefficients of degree in $\IInt{t+(k-1)j+1}{n-1}$ in  $\Lambda P^j-\Lambda R^j \mod G$ vanish, since they were affine functions by the induction hypothesis and become $0$ after the variable elimination step. This implies 
 that $\Lambda P^j-\Lambda R^j \mod G$ becomes a polynomial of degree $ \leq t+(k-1)j$ after elimination of variables.
 Therefore $P(\Lambda P^j-\Lambda R^j \mod G)$ is a polynomial of degree $\leq t+(k-1)(j+1)$.  
 From the equality of  the polynomials \eqref{term:start} and \eqref{term:end}, this implies that 
 the coefficient of degree $u$ in $\left( \Lambda P^{j+1}-\Lambda R^{j+1}\right) \mod G$ coincides with the coefficient of the same degree in $(R^j(\Lambda P-\Lambda R \mod G)) \mod G$ for 
 $u$ in $\IInt{t+(k-1)(j+1)+1}{n-1}$.
 We observe now that the last coefficient 
is nothing but a linear combination of the coefficients of $\Lambda P-\Lambda R \mod G$, which are precisely 
the initial polynomial equations. Since the polynomial  $\left( \Lambda P^{j+1}-\Lambda R^{j+1}\right) \mod G$ 
has all its coefficients that are affine functions in the $\lambda_i$'s  by Remark \ref{rem:easy} for 
all the degrees $u \in \IInt{t+(k-1)(j+1)+1}{n-1}$
 we obtain that after the Gaussian elimination step, $\cS$ contains the 
space generated by these aforementioned affine functions. This proves the proposition by induction on $j$. 
\end{proof}

These linear equations that we produce coincide exactly with the linear equations produced by the power decoding approach \cite{SSB10} and this allows to correct as many errors as the power decoding approach based on the 
same assumption, namely that they are all independent, which is actually the typical scenario. However, contrarily to power decoding that is bound to make such an assumption to work, the Gr\"obner basis is more versatile, as it allows to decode even without this assumption as explained in Section \ref{sec:experiments}.

\subsection{Decoding up to the Johnson radius}
Power decoding  \cite{SSB10} was  generalized in \cite{N18} to decode up to the Johnson radius by bringing in the ``error evaluator'' polynomial $\Omega(X)$ of degree $\leq t-1$ defined by \begin{equation}
\label{eq:interpolation}
\Omega(\supp_i)=-e_i,\;\;\text{for all $i \in \IInt{1}{n}$ for which $e_i \neq 0$.}
\end{equation}
where $e_i$ is  the error value at position $i$. In other words, it is the interpolation polynomial defined by \eqref{eq:interpolation} for all $i$ in error. This  crucially relies on \cite[Lemma 2.1]{N18}:
\begin{equation}
\label{eq:cruciallemma}
\Lambda(P-R)=\Omega G.
\end{equation}

The generalization of power decoding then uses this identity to derive further identities that are summarized by the following formulas (this is Theorem 3.1 in \cite{N18}), for any positive integer $s$ and $v$ such that 
$s \leq v$ we have
\begin{small}
\begin{align}
\Lambda^s P^u & =\sum_{i=0}^u \left( \Lambda^{s-i}\Omega^i\right) \binom{u}{i}R^{u-i}G^i & 
\text{$u \in \IInt{1}{s-1}$,} \label{eq:withoutmod}\\
\Lambda^s P^u & \equiv \sum_{i=0}^{s-1} \left( \Lambda^{s-i}\Omega^i\right) \binom{u}{i}R^{u-i}G^i \mod G^s& \text{$u \in \IInt{s}{v}$.}\label{eq:withmod}
\end{align}
\end{small}
The approach of \cite{N18} relies on the fact that when the number of errors is below the Johnson radius, there is a choice of $s$ and $v$ such that the total number of coefficients of the polynomials 
$\Lambda^s$, $\Lambda^{s-1}\Omega$, $\cdots$, $\Omega^s$ as well as $\Lambda^s P$,$\cdots$, $\Lambda^s P^v$ is less than or equal to the number of equations linking these coefficients coming from  \eqref{eq:withoutmod} and \eqref{eq:withmod}. In this case (and if these equations are independent) we recover them by solving the corresponding linear system. Notice that with this strategy, there is for a given value $s$ a maximal value for $u$ given by
$$
\pow{s} \eqdef \max\{u: st+u(k-1) \leq sn-1\}.
$$
It is readily seen that taking larger of $u$ only increase the number of variables in the linear system without 
being able to make it determinate if it was not determinate before.
Interestingly enough our Gr\"obner basis approach also exhibits degree falls of degree $s$ that are related to \eqref{eq:withoutmod} and \eqref{eq:withmod}. This can be understood by using an equivalent definition of $\Omega$ as
\begin{equation}
\label{eq:definitionOmega}
\Omega \eqdef - \Lambda R \div G.
\end{equation}
Notice that from this definition we directly derive two results
\begin{enumerate}
\item The coefficients of $\Omega$ are affine functions of the $\lambda_i$'s.
\item As long as $t \leq n-k$, \eqref{eq:cruciallemma} follows from \eqref{eq:definitionOmega} and  \eqref{eq:bilinear_better}. Indeed  $ \modulo{\Lambda R}{G}=\Lambda P$. This follows from  \eqref{eq:bilinear_better} and $t+k-1 \leq n-1$ implying that $\Lambda P = \Lambda R \mod G$. This and
 \eqref{eq:definitionOmega} then imply that
$
\Lambda R = -\Omega G + \Lambda P
$
which is obviously equivalent to \eqref{eq:cruciallemma}.
\end{enumerate}

Note that from these considerations, that if we equate the coefficients of the polynomials in  \eqref{eq:withoutmod} for all the degrees in $\IInt{st+u(k-1)+1}{st+u(n-1)}$ and in \eqref{eq:withmod} for all
the degrees in $\IInt{st+u(k-1)+1}{s(n-1)}$, the coefficient of the left-hand term vanishes and the coefficient in the righthand term is a polynomial of degree $s$ in the $\lambda_i$'s (this follows from the fact that the coefficients of $\Omega$ are affine functions in those $\lambda_i$'s).
This gives polynomial equations in the $\lambda_i$'s of degree $s$. In a sense, they can be viewed as generalizations at degree $s$ of the linear equations that were mentioned when Algorithm \ref{algo:simpleGrobner} is applied when $D=2$.
These 
equations  are actually produced as degree falls that are in the linear span of intermediate sets $\cS$ produced in Algorithm \ref{algo:simpleGrobner} when $D=s+1$. There are also other equations of degree $s$ produced by 
Algorithm \ref{algo:simpleGrobner} in such a case.
To explain this point it makes sense to bring in notation for the right-hand term in \eqref{eq:withoutmod} and \eqref{eq:withmod}. Let us define
\begin{small}
\begin{eqnarray*}
\chi(s, u)& \eqdef &\sum_{i=0}^u \binom{u}{i} \Lambda^{s-i}R^{u-i}\Omega^i G^i=
\Lambda^{s-u}\left( \Lambda R+\Omega G\right)^u \;\; \text{if } u<s,\\
\chi(s, u)&\eqdef& \modulo{\sum_{i=0}^{s-1} \binom{u}{i} \Lambda^{s-i}R^{u-i}\Omega^i G^i}{ G^{s} }\;\; \text{if } u\ge s
\end{eqnarray*}
\end{small}
We also let $\high{\chi(s,u)}$ be the polynomial where we dropped all the terms of degree $\leq ts+u(k-1)$ in $\chi(s,u)$, i.e. $\chi(s,u)=\sum_i a_i X^i$, then $\high{\chi(s,u)}=\sum_{i>ts+u(k-1)}  a_i X^i$.
\begin{theorem}
\label{th:main}
Let $\cI_D= \vect{\cS}{\F_q}$ where $\cS$ is the set output by Algorithm \ref{algo:simpleGrobner}. We have for all
nonnegative integers  $s$, $s'$, $u \leq \pow{s}$, $u' \leq \pow{s'}$
\begin{small}
\begin{eqnarray}
\high{\chi(s,u)} & \inc & \cI_{s+1} \label{eq:highcoeff} \\
\chi(s,u)\chi(s',u')-\chi(s+s',u+u')&  \inc & \cI_{s+s'+1} \label{eq:mixed}.
\end{eqnarray}
\end{small}
where  
$P \inc \cI_v$  ( where $P$ is a polynomial with coefficients that are polynomials in the $\lambda_i$'s and the $p_i$'s)  means that all the coefficients of $P$ belong to $\cI_v$.
\end{theorem}
From \eqref{eq:withoutmod} and \eqref{eq:withmod} it is of course clear that $\chi(s,u)\chi(s',u')-\chi(s+s',u+u')$
belongs to the ideal generated by the polynomial equations since they basically come from the 
identity $\Lambda^sP^u\Lambda^{s'}P^{u'}=\Lambda^{s+s'}P^{u+u'}$. What is somehow surprising is that these equations are actually 
discovered at a rather small degree Gr\"obner basis computation (namely by staying at degree $s+s'+1$). Moreover these
equations only involve the $\lambda_i$'s. By inspection of the behavior of the Gr\"obner basis computation, it seems that the linear equations that we produce later on are first produced by degree falls only involving these equations of degree $s$. It is therefore tempting to change the Gr\"obner basis decoding procedure strategy: 
instead of feeding Algorithm \ref{algo:simpleGrobner} with the initial system \eqref{eq:bilinear} or \eqref{eq:bilinear_better} we run it with the equations of degree $s$ given by Theorem \ref{th:main}. Once we have recovered the $\lambda_i$'s in this way we recover the $p_i$'s by solving a linear system as explained earlier. How this strategy behaves on non-trivial examples is now explained in the next section.

This result is proved in the following subsection.
\subsection{Proof of Theorem \ref{th:main}}

It will be convenient here to notice that $\chi(s,s)$ has a slightly simpler expression which avoids the reduction modulo $G^s$. 
\begin{lemma}\label{lem:chiss}
$$\chi(s,s) = (\Lambda R + \Omega G)^s.$$
\end{lemma}

\begin{proof}
$\chi(s,s)$ is defined as 
\begin{eqnarray*}
\chi(s,s) &\eqdef &\modulo{ \sum_{i=0}^{s-1} \binom{s}{i} \Lambda^{s-i}R^{s-i} \Omega^i G^i}{G^s}\\
& = & \modulo{ \sum_{i=0}^{s} \binom{s}{i} \Lambda^{s-i}R^{s-i} \Omega^i G^i}{G^s} \\
&= & \modulo{(\Lambda R+\Omega G)^s}{G^s}\\
&= & (\Lambda R+\Omega G)^s
\end{eqnarray*}
\end{proof}

It will also be helpful to observe that $\chi(s,u)$ and $\chi(s,u+1)$ are related by the following identity
\begin{lemma}
\label{lem:identity}
\begin{eqnarray*}
\chi(s,u)P-\chi(s,u+1) &= & \Lambda^{s-u-1}(\Lambda R+\Omega G)^u \left( \Lambda P - \Lambda R - \Omega G\right) \;\; \text{for $u \in \IInt{0}{s-1}$}\\
\modulo{\chi(s,u)P-\chi(s,u+1)}{G^s} & = & \modulo{\left(\Lambda P-\Lambda R-\Omega G \right)\sum_{i=0}^{s-1} \binom{u}{i} \Lambda^{s-1-i}R^{u-i}\Omega^i G^i}{G^s} \;\; \text{for $u \in \IInt{s}{\pow{s}-1}$.}
\end{eqnarray*}
\end{lemma}
\begin{proof}
For $u \in \IInt{0}{s-1}$ we have (for the case $u=s-1$  we use Lemma \ref{lem:chiss} for the term
$\chi(s,u+1)$):
\begin{eqnarray*}
\chi(s,u)P-\chi(s,u+1) &= &\Lambda^{s-u}P(\Lambda R + \Omega G)^u -\Lambda^{s-u-1}P(\Lambda R + \Omega G)^{u+1} \\
& = & \Lambda^{s-u-1}(\Lambda R+\Omega G)^u \left( \Lambda P - \Lambda R - \Omega G\right).
\end{eqnarray*}
For $u \in \IInt{s}{\pow{s}-1}$ we observe that 
\begin{eqnarray}
\modulo{\chi(s,u)P}{G^s} & = &  \modulo{P \sum_{i=0}^{s-1} \binom{u}{i} \Lambda^{s-i}R^{u-i}\Omega^i G^i}{ G^{s} } \nonumber\\ &= &\modulo{\Lambda P \sum_{i=0}^{s-1} \binom{u}{i} \Lambda^{s-1-i}R^{u-i}\Omega^i G^i }{G^s}\label{eq:chisu}
\end{eqnarray}
and
\begin{eqnarray*}
(\Lambda R+ \Omega G)\sum_{i=0}^{s-1} \binom{u}{i} \Lambda^{s-1-i}R^{u-i}\Omega^i G^i
& = & \sum_{i=0}^{s-1} \binom{u}{i} \Lambda^{s-i}R^{u+1-i}\Omega^i G^i +
\sum_{i=0}^{s-1} \binom{u}{i} \Lambda^{s-1-i}R^{u-i}\Omega^{i+1} G^{i+1}\\
& = & \Lambda^s R^{u+1} + \binom{u}{s-1} R^{u-s+1}\Omega^sG^s \\
& & +
\sum_{i=1}^{s-1} \left( \binom{u}{i}+ \binom{u}{i-1}\right)\Lambda^{s-i}R^{u+1-i}\Omega^i G^i\\
& = &\binom{u}{s-1}R^{u-s+1}\Omega^sG^s+ \sum_{i=0}^{s-1} \binom{u+1}{i}\Lambda^{s-i}R^{u+1-i}\Omega^i G^i
\end{eqnarray*}
This implies
\begin{equation}
\label{eq:chisupu}
\chi(s,u+1) = \modulo{(\Lambda R+ \Omega G)\sum_{i=0}^{s-1} \binom{u}{i} \Lambda^{s-1-i}R^{u-i}\Omega^i G^i}{G^s}.
\end{equation}
The second equation of the lemma follows directly from \eqref{eq:chisu} and \eqref{eq:chisupu}.
\end{proof}

A last lemma will be helpful now
\begin{lemma}
\label{lem:main}
For all
nonnegative integers $s$ and  $u <\pow{s}$
\begin{eqnarray}
\chi(s,u)P - \chi(s,u+1)& \inc& \cI_{s+1} \label{eq:basic}\\
\high{\chi(s,u+1)} & \inc & \cI_{s+1}. \label{eq:high}
\end{eqnarray}
\end{lemma}
\begin{proof}
We will prove this lemma by induction on $u$. For $u \leq s-1$ we observe from Lemma \ref{lem:identity} that
\begin{eqnarray}
\chi(s,u)P-\chi(s,u+1) &= &
\Lambda^{s-u-1}(\Lambda R+\Omega G)^u \left( \Lambda P - \Lambda R - \Omega G\right) \label{eq:direct}\\
& \inc & \cI_{s+1} \nonumber
\end{eqnarray}
The last point follows from the fact that \eqref{eq:direct} implies that the coefficients of $\chi(s,u)P-\chi(s,u+1)$ 
are clearly in the space spanned by $\cS$ once we multiply the original $f_i$'s (i.e. the coefficients of 
$\Lambda P-\Lambda R-\Omega G$) by all monomials of degree $\leq s-1$) because the coefficients 
of $\Lambda^{s-u-1}(\Lambda R+\Omega G)^u$ are polynomials of degree $\leq s-1$ in the $\lambda_i$'s.

This also implies that $\high{\chi(s,u+1)} \inc \mathcal{I}_{s+1}$, since $\deg{\chi(s,u)P}=ts+u(k-1)$.
Now let us assume that $\chi(s,u-1)P-\chi(s,u)\inc \mathcal{I}_{s+1}$ and $\chi(s,u)_H\inc\mathcal{I}_{s+1}$, for some $s \leq u< \pow{s}$. From Lemma \ref{lem:identity} we know that
$$
\modulo{\chi(s,u)P - \chi(s,u+1)}{G^s}= \modulo{\left(\Lambda P-\Lambda R-\Omega G \right)\sum_{i=0}^{s-1} \binom{u}{i} \Lambda^{s-1-i}R^{u-i}\Omega^i G^i}{G^s}.
$$
Therefore 
$$
\modulo{\chi(s,u)P - \chi(s,u+1)}{G^s} \inc \cI_{s+1}
$$
since clearly (i)
$$\left(\Lambda P-\Lambda R-\Omega G \right)\sum_{i=0}^{s-1} \binom{u}{i} \Lambda^{s-1-i}R^{u-i}\Omega^i G^i
\inc  \cI_{s+1}$$
$\left(\Lambda P-\Lambda R-\Omega G \right)\sum_{i=0}^{s-1} \binom{u}{i} \Lambda^{s-1-i}R^{u-i}\Omega^i G^i$.

By the induction hypothesis $\high{\chi(s,u)} \inc \mathcal{I}_{s+1}$ and such coefficients have degree $s$, then the coefficients corresponding to degrees $>ts+(u+1)(k-1)$ of $\chi(s,u)P$ belong to $\mathcal{I}_{s+1}$ too. 
Since $\modulo{\chi(s,u)P - \chi(s,u+1)}{G^s}= \modulo{\chi(s,u)P }{G^s}- \chi(s,u+1)$, it follows that 
\[
\chi(s,u)P-\chi(s,u+1)\in \mathcal{I}_{s+1}.
\]
Thus, we also have $\chi(s,u+1)_H\in \mathcal{I}_{s+1}$.
\end{proof}

We are ready now to prove Theorem \ref{th:main}.
\begin{proof}[Proof of Theorem \ref{th:main}]
We proceed by induction on $u_1$ and $u_2$. We first observe that we trivially have
$\chi(s_1,0)\chi(s_2,0)-\chi(s_1+s_2,0) \inc \cI_{s_1+s_2+1}$ since
$$
\chi(s_1,0)\chi(s_2,0)-\chi(s_1+s_2,0) = \Lambda^{s_1}\Lambda^{s_2}-\Lambda^{s_1+s_2}=0.
$$
Now assume that we have
$$
\chi(s_1,u_1)\chi(s_2,u_2)-\chi(s_1+s_2,u_1+u_2) \inc \cI_{s_1+s_2+1},
$$
for some positive integers $s_1$ and $s_2$ and non-negative integers $u_1 < \pow{s_1}$ and
$u_2 \leq \pow{s_2}$. Since $\chi(s_1,u_1)\chi(s_2,u_2)$ and $\chi(s_1+s_2,u_1+u_2)$ are polynomials where all coefficients are polynomials in the $\lambda_i$'s of degree $\leq s_1+s_2$,  we also have 
\begin{equation}
\label{eq:multbyP}
P \left( \chi(s_1,u_1)\chi(s_2,u_2)-\chi(s_1+s_2,u_1+u_2)\right)  \inc \cI_{s_1+s_2+1}.
\end{equation}

By Lemma \ref{lem:main} we know that
$$
P \chi(s_1,u_1)-\chi(s_1,u_1+1) \inc \cI_{s_1+1}.
$$
This implies
\begin{equation}
\label{eq:chisuuu}
P \chi(s_1,u_1)\chi(s_2,u_2)-\chi(s_1,u_1+1)\chi(s_2,u_2) \inc \cI_{s_1+s_2+1}.
\end{equation}
On the other hand, still by Lemma \ref{lem:main}, we have
\begin{equation}
\label{eq:chissum}
P \chi(s_1+s_2,u_1+u_2)-\chi(s_1+s_2,u_1+u_2+1) \inc \cI_{s_1+s_2+1}.
\end{equation}
From \eqref{eq:chisuuu} and \eqref{eq:chissum} we derive that
\begin{equation}\label{eq:sum}
-P \chi(s_1,u_1)\chi(s_2,u_2)+\chi(s_1,u_1+1)\chi(s_2,u_2) 
+ P \chi(s_1+s_2,u_1+u_2)-\chi(s_1+s_2,u_1+u_2+1) \inc  \cI_{s_1+s_2+1}
\end{equation}
\eqref{eq:sum} and \eqref{eq:multbyP} imply that
$$
\chi(s_1,u_1+1)\chi(s_2,u_2) - \chi(s_1+s_2,u_1+u_2+1) \inc  \cI_{s_1+s_2+1}.
$$
This proves the theorem by induction (the induction on $u_2$ follows directly from the fact we can exchange the role of $u_1$ and $u_2$).
\end{proof}

\section{Experimental Results}\label{sec:experiments}
In this section, we compare the behavior of a $D$-Gr\"obner basis
computation on the bilinear system~\eqref{eq:bilinear_better}, with a
system involving equations in the $\lambda_j$'s only. We give
examples where Johnson's bound is attained and passed.

The systems in $\lambda_j$'s we use contains equations $\chi(s,u)_H$
and some relations
$\chi(s,u)\chi(s',u')-\chi(s+s',u+u')$. Experimentally, they are
linearly dependent from $\chi(s+s'-1,u+u')\chi(1,0)-\chi(s+s',u+u')$
and $\chi(s,\pow{s})_H$. Moreover,
$ \chi(s-1,u)\chi(1,0) \mod G^{s-1}=\chi(s,u) \mod G^{s-1}$,
so we will  consider equations $\mathcal M_{s,u}$ defined by
\begin{align}
\left(\chi(s-1, u)\chi(1,0)-\chi(s, u)\right) \div {G^{s-1}}.\tag{$\mathcal M_{s,u}$}
\end{align}
We do not add equations that are polynomially dependent from
$\chi(s,\pow{s})_H$ or $\mathcal M_{s+1,\pow{s}}$ at degree at most $D$,
and thus unnecessary for the computation.

Tables~\ref{tab:n64},~\ref{tab:n256} and~\ref{tab:n37} show results for $[n,k]_q$ taking values $[64,27]_{64}$, $[256,63]_{256}$ and
$[37,5]_{61}$.  The column $\#\lambda_j$ indicates the number of
remaining $\lambda_j$'s after elimination of the linear ones from the
$\chi(1,*)_H$ relations. The column ``Eq'' indicates the equations
used. The column ``\#Eq'' contains the
degrees of the equations\footnote{2:45 means that the system contains
  45 equations of degree 2.}.

We do our experiments using the \texttt{GroebnerBasis(S,D)} function
in the computer algebra system \texttt{magma v2.25-6}. The practical
complexity $\mathbb C$ is given by the \texttt{magma} function
\texttt{ClockCycles}. For instance, on our machine with an 
Intel\textsuperscript{\textregistered}
Xeon\textsuperscript{\textregistered} 2.00GHz processor, $2^{30.9}$
clock cycles are done in 1 second, $2^{36.8}$ in 1 minute and
$2^{42.7}$ in 1 hour.  ``Max Matrix'' indicates the size of the
largest matrix during the process.  The complexities include the computation of the
equations $\chi(i,j)_H$ and $\mathcal M_{i,j}$ that could be improved.

For systems where the number of remaining $\lambda_j$'s is small
compared to the number of $p_i$'s, e.g. Table~\ref{tab:n64} or
Table~\ref{tab:n256}, it is clearly interesting to compute a Gröbner
basis for a system containing only polynomials in $\lambda_j$'s: even
if the maximal degree $D$ is larger than for the bilinear system, the
number of variables is much smaller and the computation is faster. For
instance for $[n,k]_q=[64, 27]_{64}$ in Table~\ref{tab:n64}, on
Johnson bound $t=23$ the Gröbner basis for the bilinear system
requires more than 6 hours of computation and 47 GB of memory, whereas
the computation in $\lambda_j$'s only takes less than a second. For
$t=24$ we couldn't solve the bilinear system directly, whereas the
system in $\lambda_j$'s only solves in less than a minute.

\setlength\tabcolsep{2pt}
\begin{table}[htbp]
  \caption{Experimental results for a $[n, k]_q = [64, 27]_{64}$
    RS-code. System~\eqref{eq:bilinear_better} contains 26 variables
    $p_i$. Johnson's bound is $t=23$.}
  \label{tab:n64}
  \centering
  \begin{tabular}[c]{|c|@{}c@{}|*{7}{c|}}
    \hline
    $t$ &  $\# \lambda_j$ & Eq. & \#Eq. & $D$ & Max Matrix & $\mathbb C$ \\\hline
    $19$ & $1$ & \eqref{eq:bilinear_better} & 2:45 & 2 & $65 \times 57$ & $2^{22.2}$\\
        && $\chi(2,3)_H$ & 2:11 & 2 & $45 \times 28$ & $2^{23.7}$\\\hline
    $20$ & $3$ & \eqref{eq:bilinear_better} & 2:46 & 3 & $1522 \times 1800$ & $2^{26.5}$\\
        && $\chi(2,3)_H$& 2:9 & 2 & $47 \times 28$ & $2^{24.4}$ \\\hline
    $21$ & $5$ & \eqref{eq:bilinear_better} & 2:47 & 3 & $1711 \times 2889$ &$2^{27.1}$\\
        &&$\chi(2,3)_H$ + $\chi(3,4)_H$ & 2:7, 3:24  & 3 & $66 \times 56$ & $2^{26.8}$ \\\hline
    $22$ & $7$ & \eqref{eq:bilinear_better} & 2:48 & 4 & $31348 \times 35972$ & $2^{36.1}$\\
        &&$\chi(2,3)_H$ + $\chi(3,4)_H$ & 2:5, 3:21  & 4 & $271 \times 283$ & $2^{27.6}$\\\hline
    $23$ & $9$ & \eqref{eq:bilinear_better} &  2:49 & 5 & $428533 \times 406773$ & $2^{45.4}$\\
        && $\chi(2,3)_H$ + $\mathcal M_{3, 3}$ & 2:4, 3:22 & 5 & $1466 \times 1641$ & $2^{30.1}$\\\hline
    $24$ & $11$ & \eqref{eq:bilinear_better} & 2:50  & $\ge 6$ & -- & -- \\
        && $\mathcal M_{3, 3}$ & 2:1, 3:23 & 7 & $28199 \times 23536$ & $2^{35.8}$\\\hline
  \end{tabular}
\end{table}

Table~\ref{tab:n256} gives an example where the number of
$\lambda_j$'s variables is quite large, but still smaller than the
number of $p_i$'s. The benefit of using equations in $\lambda_j$'s
only is clear.

\begin{table}[htbp]
  \caption{Experimental results for a $[n, k]_q = [256, 63]_{256}$
    RS-code. System~\eqref{eq:bilinear_better} contains 62 variables
    $p_i$. Johnson's bound is $t=130$.}
  \label{tab:n256}
  \centering
  \begin{tabular}[c]{|c|@{}c@{}|p{1.7cm}|p{2.0cm}|*{7}{c|}}
    \hline
    $t$ &  $\# \lambda_j$ & Eq. & \#Eq. & $D$ & Max Matrix & $\mathbb C$ \\\hline
    $120$ & $36$ & \eqref{eq:bilinear_better} & 2:182 & 3 &$20023 \times 128018$ & $2^{38.0}$\\
    && $\chi(2,3)_H$ & 2:85 & 2 &$119 \times 703$ & $2^{34.5}$\\\hline
    $121$ & $39$ & \eqref{eq:bilinear_better} & 2:183 & 3 &  $21009 \times 143741$ & $2^{38.9}$\\
        && $\mathcal M_{2,2}$& 2:111 & 3 & $9780 \times 8517$ &  $2^{35.0}$\\\hline
    $122$ & $42$ & \eqref{eq:bilinear_better} & 2:184 & 3 &  $22050 \times 160434$ &$2^{39.7}$\\
        && $\mathcal M_{2,2}$&2:113  & 3 & $4858 \times 14189$ & $2^{35.3}$ \\\hline
    $123$ & $45$ & \eqref{eq:bilinear_better} & 2:185 & 3 & $23112 \times 178090$  & $2^{40.1}$\\
        & & $\mathcal M_{2,2}$ & 2:115 & 3 & $5289 \times 17295$ & $2^{35.8}$\\\hline
    $124$ & $48$ & \eqref{eq:bilinear_better} &  2:186 & $\ge 4$ & -- & -- \\
             &&  $\mathcal M_{2,2}+ \mathcal M_{4,6}$  & 2:117, 3:1, 4:189 & 4 & $164600 \times 270725$ & $2^{45.2}$\\\hline
  \end{tabular}
\end{table}

On the contrary, Table~\ref{tab:n37} shows that for a small value of
$k$ compared to the number of $\lambda_j$'s, the maximal degree for
the bilinear system is smaller than the one for a system involving only $\lambda_j$'s, but the total number of variables is almost the same, hence it is more
interesting to solve directly the bilinear system. Moreover, here computing the $\mathcal M_{i,j}$ equations (that are equations in $\lambda_j$'s of degree $i$) takes time. Note that, for
$t\ge 26$ we may have several solutions: the Gr\"obner basis
computation performs a list decoding and returns all the solutions.

\begin{table}[htbp]
  \caption{Experimental results for a $[n, k]_q = [37, 5]_{61}$
    RS-code. System~\eqref{eq:bilinear_better} contains 4 variables
    $p_i$. Johnson's bound is $t=24$, Gilbert-Varshamov's bound is
    $t=28$. 
  }
\label{tab:n37}
  \centering
  \begin{tabular}[c]{|c|@{}c@{}|p{2.5cm}|p{1.2cm}|*{7}{c|}}
    \hline
    $t$ &  $\# \lambda_j$ & Eq. & \#Eq. & $D$ & Max Matrix & $\mathbb C$ \\\hline
    $24$ & $12$ & \eqref{eq:bilinear_better} & 2:28 & 3 & $1065 \times 1034$ & $2^{26.0}$\\
    && $\mathcal M_{2, 3}$ & 2:37 & 3 & $454 \times 454$ & $2^{28.0}$\\\hline
    $25$ & $15$ & \eqref{eq:bilinear_better} & 2:29 & 3 & $2520 \times 1573$ & $2^{28.0}$\\
    && $\chi(2,5)_H$+ $\chi(3,8)_H$+ $\mathcal M_{2,2}$ + $\mathcal M_{3,5}$& 2:25, 3:40 & 4 & $3193 \times 3311$ & $2^{34.3}$ \\\hline
    $26$ & $18$ & \eqref{eq:bilinear_better} & 2:30 & 4 & $20446 \times 15171$ &$2^{33.1}$\\
    &&$\chi(2,5)_H$+ $\mathcal M_{2,2}$ + $\mathcal M_{3, 5}$ + $\mathcal M_{4, 8}$ &2:25, 3:37, 4:37  & 5 & $38796 \times 22263$ & $2^{38.1}$ \\\hline
    $27$ & $21$ & \eqref{eq:bilinear_better} & 2:31 & 4 & $27366 \times 24894$ & $2^{36.0}$\\
\hline
  \end{tabular}
\end{table}

\section{Concluding remarks}
This paper demonstrates that using a standard Gr\"obner basis computation on the bilinear system \eqref{eq:bilinear_better} for decoding a Reed-Solomon code is of polynomial complexity below  Sudan's radius.
The Gr\"obner basis computation reveals polynomial equations of small degree involving the coefficients $\lambda_i$  of the error locator polynomial. They are related to the  power decoding approach \cite{N18}. We give a theorem explaining why these  polynomial relations are obtained at a surprisingly small degree. This is a first step for understanding why the Gr\"obner works
surprisingly well beyond the Sudan radius and is successful  by staying at a small degree. We have also explored another way of using this approach, namely by feeding some of the aforementioned polynomial relations directly in a Gr\"obner basis computation. This results in some cases in a considerable complexity gain. We have considered some of the examples given in \cite{N18} and show that this Gr\"obner basis approach can still be effective slightly beyond Johnson's bound. We also remark that contrarily to the power decoding approach which is restricted to the case where there is a unique solution to the decoding problem,  the Gr\"obner basis computation is also able to compute all solutions. This approach opens new roads for decoding algebraically a Reed-Solomon code.


\begin{thebibliography}{CLO15}

\bibitem[AK11]{MM11}
Mortuza Ali and Margreta Kuijper.
\newblock A Parametric Approach to List Decoding of
                  Reed-Solomon Codes Using Interpolation.
\newblock {\em IEEE Trans. Inf. Theory}, 57(10):6718-6728, 2011.

\bibitem[Bar04]{B04}
Magali Bardet.
\newblock {\em {\'E}tude des syst{\`e}mes alg{\'e}briques
  surd{\'e}termin{\'e}s. Applications aux codes correcteurs et {\`a} la
  cryptographie}.
\newblock PhD thesis, Universit{\'e} Paris VI, December 2004.
\newblock http://tel.archives-ouvertes.fr/tel-00449609/en/.

\bibitem[BW18]{BW18} Hannes Bartz and Antonia Wachter-Zeh.
\newblock Efficient decoding of interleaved subspace and
                  Gabidulin codes beyond their unique decoding radius
                  using Gröbner bases.
\newblock {\em Advances in Mathematics of Communications}, 12(4):773-804, 2018.

\bibitem[CLO15]{CLO15}
David Cox, John Little, and Donal O'Shea.
\newblock {\em Ideals, Varieties, and algorithms: an Introduction to
  Computational Algebraic Geometry and Commutative Algebra.}
\newblock Undergraduate Texts in Mathematics, Springer-Verlag, New York., 2015.

\bibitem[Fau99]{F99}
Jean-Charles Faug\`{e}re.
\newblock A new efficient algorithm for computing {Gr{\"{o}}bner} bases ({F4}).
\newblock {\em J. Pure Appl. Algebra}, 139(1-3):61--88, 1999.

\bibitem[FSS11]{FSS11}
Jean-Charles {Faug{\`e}re}, Mohab {Safey El Din}, and Pierre-Jean
  {Spaenlehauer}.
\newblock {Gr\"obner} bases of bihomogeneous ideals generated by polynomials of
  bidegree (1,1): Algorithms and complexity.
\newblock {\em J. Symbolic Comput.}, 46(4):406--437, 2011.

\bibitem[Gao03]{G03a}
Shuhong Gao.
\newblock A new algorithm for decoding {Reed-Solomon} codes.
\newblock In Vijay~K. Bhargava, H.~Vincent Poor, Vahid Tarokh, and Seokho Yoon,
  editors, {\em Communications, Information and Network Security}, pages
  55--68, Boston, MA, 2003. Springer US.

\bibitem[GS98]{GS98}
Venkatesan Guruswami and Madhu Sudan.
\newblock Improved decoding of {R}eed--{S}olomon and algebraic-geometric codes.
\newblock In {\em Proceedings 39th Annual Symposium on Foundations of Computer
  Science (Cat. No. 98CB36280)}, pages 28--37. IEEE, 1998.

\bibitem[Laz83]{L83}
D.~Lazard.
\newblock Gr{\"o}bner bases, {G}aussian elimination and resolution of systems
  of algebraic equations.
\newblock In {\em Computer algebra}, volume 162 of {\em LNCS}, pages 146--156,
  Berlin, 1983. Springer.
\newblock Proceedings Eurocal'83, London, 1983.

\bibitem[LO06]{LO06b}
  Kwankyu Lee and Michael E. O'Sullivan.
  \newblock An Interpolation Algorithm using Gröbner Bases for
                  Soft-Decision Decoding of Reed-Solomon Codes.
  \newblock In {\em 2006 IEEE International Symposium on Information
                  Theory}, pages 2032-2036, 2006.

\bibitem[LO08]{LO08}
  Kwankyu Lee and Michael E. O'Sullivan.
  \newblock List decoding of Reed–Solomon codes from a Gröbner
                  basis perspective.
  \newblock	 {\em Journal of Symbolic Computation}, 43(9):645-658, 2008.


\bibitem[Nie14]{N14}
Johan Sebastian~Rosenkilde Nielsen.
\newblock Power decoding of reed-solomon codes revisited.
\newblock In Raquel Pinto, Paula~Rocha Malonek, and Paolo Vettori, editors,
  {\em Coding Theory and Applications, 4th International Castle Meeting,
  {ICMCTA} 2014, Palmela Castle, Portugal, September 15-18, 2014}, volume~3 of
  {\em {CIM} Series in Mathematical Sciences}, pages 297--305. Springer, 2014.

\bibitem[Nie18]{N18}
Johan Sebastian~Rosenkilde Nielsen.
\newblock Power decoding {Reed-Solomon} codes up to the {Johnson} radius.
\newblock {\em Advances in Mathematics of Communications}, 12(1):81, 2018.

\bibitem[Spa12]{S12}
Pierre{-}Jean Spaenlehauer.
\newblock {\em R{\'e}solution de syst{\`e}mes multi-homog{\`e}nes et
  determinantiels}.
\newblock PhD thesis, Univ. Pierre et Marie Curie- Paris 6, October 2012.

\bibitem[SSB10]{SSB10}
Georg Schmidt, Vladimir Sidorenko, and Martin Bossert.
\newblock Syndrome decoding of {Reed-Solomon} codes beyond half the minimum
  distance based on shift-register synthesis.
\newblock {\em {IEEE} Trans. Inf. Theory}, 56(10):5245--5252, 2010.

\bibitem[Sud97]{S97b}
Madhu Sudan.
\newblock Decoding of {Reed--Solomon} codes beyond the error--correction bound.
\newblock {\em J. Complexity}, 13(1):180--193, 1997.

\bibitem[Tri10]{T10}
Peter V. Trifonov.
\newblock Efficient Interpolation in the Guruswami–Sudan
                  Algorithm.
\newblock {\em IEEE Trans. Inf. Theory}, 56(9):4341-4349, 2010.

\bibitem[ZS10]{ZS10}
Alexander Zeh and Christian Senger.
\newblock A link between Guruswami-Sudan's list-decoding and
                  decoding of interleaved Reed-Solomon codes.
\newblock In {\em 2010 IEEE International Symposium on Information
                  Theory}, pages 1198-1202, 2010.

\end{thebibliography}
\end{document}